\documentclass[letterpaper, 10 pt, conference]{ieeeconf}  %
\newcommand{\RN}[1]{%
\textup{\uppercase\expandafter{\romannumeral#1}}%
}
\newcommand{\norm}[1]{\left\lVert{#1}\right\rVert}
\IEEEoverridecommandlockouts                              
\usepackage{graphicx}                                                          
\usepackage{nomencl}
\usepackage{amsmath}
\usepackage{amssymb}
\usepackage{MnSymbol}
\usepackage{epstopdf}
\usepackage{array,tabularx}
\newtheorem{theorem}{Theorem}
\makenomenclature

\title{\LARGE \bf
A geometric approach to the dynamics of flapping wing micro aerial vehicles: Modelling and reduction
}

\author{S. Kadam$^{1}$, S. Gajbhiye$^{1}$ and R. N. Banavar$^{1}$
\thanks{$^{1}$
S Kadam, S Gajbhiye, R N Banavar are with Systems and Control Engineering Department, Indian Institute of Technology Bombay,
        Mumbai, 400076, India
        {\tt\small sudin,sneha@sc.iitb.ac.in, banavar@iitb.ac.in}}%
}

\begin{document}
\maketitle
\thispagestyle{empty}
\pagestyle{empty}
\begin{abstract}
This paper presents a geometric framework for analysis of dynamics of flapping wing micro aerial vehicles (FWMAV) which achieve locomotion in the special Euclidean group SE(3) using internal shape changes. We review the special structure of the configuration manifold of such systems. This work addresses to extend the work in geometric locomotion to the aerial locomotion problem. Furthermore, there seems to be limited work in modelling of flapping wing bodies in a geometric framework. We derive the dynamic model of the FWMAV using Lagrangian reduction theory defined on symmetry groups. The reduction is achieved by applying Hamilton's variation principle on a reduced Lagrangian. The resultant dynamics is governed by the Euler-Poincar\'{e} and Euler-Lagrange equations.

\end{abstract}

\section{INTRODUCTION}

Locomotion of robots or living organisms relates to body movements that results in progression from one place to another in a physical space. The locomotion of articulated mechanical systems is often complex, even when considered with the aid of reduction principles. Geometric mechanics finds a very important position in analysis of robotic locomotion \cite{Bloch book}. For a large class of locomotion systems, including underwater vehicles, spacecraft with rotors and wheeled or legged robots, it is possible to model the motion of the system using the geometric phase associated with a connection on a principal bundle  \cite{liang}, \cite{cabrera}, \cite{marsden_krishna_bloch}. Understanding the locomotion of animals and robots involves nonlinear dynamics and the coordination of multiple limbs. A commonality in most of the approaches in locomotion is the periodic variation in limbs or shape variable to achieve the macro-motion. This idea of using periodic driving signals to produce macro-movement appears in a number of settings as explained in \cite{liang}, \cite{cabrera}, \cite{Fairchild beanie}, \cite{Chung dorothy}, \cite{Kuang robobat}.

A flapping wing vehicle is an aerial body like a bird or an insect that uses periodic wing motions for thrust as well as lift generation. The idea of aerial robotic locomotion is to understand and imitate the flapping-wing flight of birds, insects etc. Agility, vertical take-off and landing (VTOL) ability distinguishes it from conventional aircraft, making it suitable for various applications. Work in this area has been a topic of research since long, and many systems are engineered and commercially used which mimic natural aerial locomotion \cite{Taha thesis}, \cite{paranjape review paper}. Survey papers \cite{paranjape review paper}, \cite{Taha review} highlight that bio-inspired flight is receiving attention in the aerospace and robotics research community. A few engineering bird-like MAVs have been commercially developed such as Festo SmartBird and the Aerovironment Hummingbird in the recent years \cite{paranjape review paper}.

\begin{figure}[h!]
\includegraphics[scale=0.5]{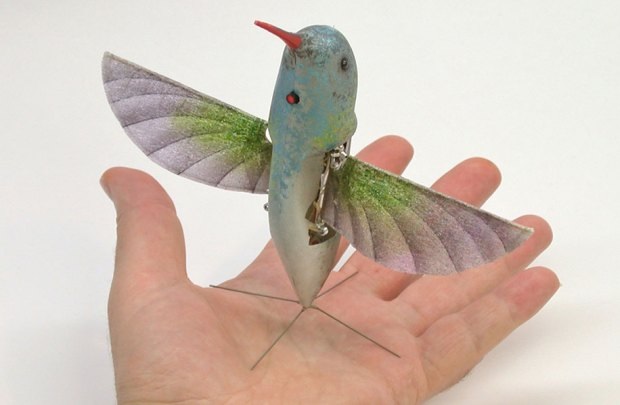} 
\label{hummingbird}
\caption{Hummingbird}
\end{figure}

From the literature on the topic it is noticed that the modelling of dynamics of the FWMAV is done using analytical mechanics and local parametrizations, and the particular geometric structure for locomoting bodies is not utilized in the setting of FWMAVs. It appears that power of tools from geometric mechanics and control remains to be applied to the aerial locomotion problem. The purpose of this paper is to explore the application of tools from geometric mechanics and control to a problem on a more general structure group - $SE(3)$ to the FWMAV. Here we refer to \cite{Morgansen UUV} which presents geometric methods for modelling and control of free-swimming fin-actuated underwater vehicles. This seems to be one such work analyzing geometry of locomotion on $SE(3)$, although for a hydrodynamic environment. Capturing the aerodynamics of a FWMAV poses lot of complexities resulting from nonlinear, unsteady effects, to list a few. There are many models proposed of various degrees of complexity in the literature \cite{Chung dorothy},\cite{Taha thesis},\cite{Taha Hajj},\cite{Kuang robobat},\cite{paranjape review paper}. In the present work we focus on the modelling and analysis of rigid body dynamic aspect of the system along with qualitative inclusion of the aerodynamics forces, rather than delving into the source of aerodynamic effects and their mathematical forms.

While studying locomotion using internal shape change, the topology of the configuration space needs attention \cite{kelly murray}. The configuration space is divided into two parts. One part describes the configuration of the internal shape variables of the mechanism, referred to as the shape or base manifold $B$. The other part is the macro-position of the robot, a Lie group $G$, representing displacement of the body coordinate frame with respect to the reference frame. The total configuration space of the robot, $Q$ is defined by both $G$ and $B$. Such systems follow the topology of a trivial principal fiber bundle, defined mathematically in Appendix A. The fibers over each point of the base correspond to the different values of the group variables. The base manifold $B$ is the quotient (Appendix B) of the entire configuration manifold $Q$ with the group variables as the equivalence class. This paper approaches the locomotion problem for FWMAV in such a setting. The Fig. \ref{fiber_bundle} shows a schematic of a fiber bundle \cite{wolfram_website}. We refer to \cite{Gallier}, \cite{Wilson} for more details on the topology of locomoting systems .

\begin{figure}[h!]
\centering
\includegraphics[scale=0.35]{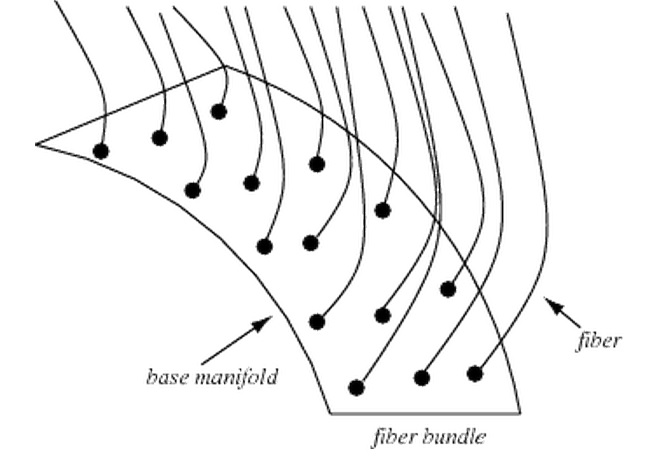} 
\caption{Fiber Bundle}
\label{fiber_bundle}
\end{figure}

In this paper we use the ideas of geometric mechanics to derive the equations of motion for the FWMAV. While some classical approaches like Euler-Lagrange equation lead to equations of motion, the geometric mechanics framework helps to understand the structure and geometric properties of the system. The classical procedure for deriving dynamics is by Hamilton's variations and the resultant equations of motion are Euler-Lagrange equation. Continuous symmetries are a vital aspect of geometric mechanics since they facilitate in lowering the dimension of the system by constructing the corresponding reduced system. In the geometric approach we consider the configuration space as a Lie group structure and identify some symmetry in the system which leads to reduced dynamics. If the Lagrangian is invariant under the Lie group $G$, the group is termed as \textit{symmetry} group. Exploiting the symmetries in the system leads to dynamical model on reduced space, which is generally the shape space for a locomotion problem. The reduction is based on identifying configuration space as Lie group and factoring the dependence of symmetry, then applying Hamilton's variation principle which leads to Euler-Poincar\'{e} equations on reduced space. The relevant work on the reduction procedure are discussed in \cite{marsden_krishna_bloch}, \cite{ostrowski}, \cite{cendra}. Reduction becomes powerful in the context of control theory as mostly locomotion systems have full control of the shape variables. This approach also seems well suited for studying issues of controllability and choice of gait as detailed in \cite{kelly murray}. Sometime the $G$-invariance is broken by gravity or by nonholonomic constraint, leaving symmetry with respect to a subgroup of $G$. Hence, the resultant equations of motion are Euler-Poincar\'{e} with an advected parameter \cite{cendra}, \cite{holm}, \cite{gajbhiye_banavar}. In this paper we construct the dynamics using Lagrangian reduction theory defined on a symmetry group.

The rest of the paper is organized as follow. In the following section we explain the geometry of the wing motion followed by definition of the configuration space. In section $\RN{3}$ we derive the expression for the Lagrangian of the system and explain the Lie group symmetry for FWMAV, followed by the derivation of reduced dynamic equations including the aerodynamic forces. Section $\RN{4}$ concludes the paper and lays out future scope of work.

\section{Flapping Winged vehicle}
A flapping wing vehicle uses complex wing motions to generate aerodynamic lift and thrust to manoeuvre in 3 dimensional space. There can be multiple variants of a  flapping wing mechanism. In the following part we explain the geometry of the FWMAV analysed in this work.

\subsection{Description}
Fig. \ref{fig:Flapping winged vehicle schematic} shows a schematic of the MAV whose dynamics is modelled in this paper. The body is interchangeably referred to as the torso. Two rigid wing are joined to the body at its center of mass such that wings can have rotation about any axis passing through this joint. The wing actuation occurs through torques applied at these two joints.

\vspace{10pt}
\begin{figure}[!htb]
\centering
\includegraphics[scale=.35]{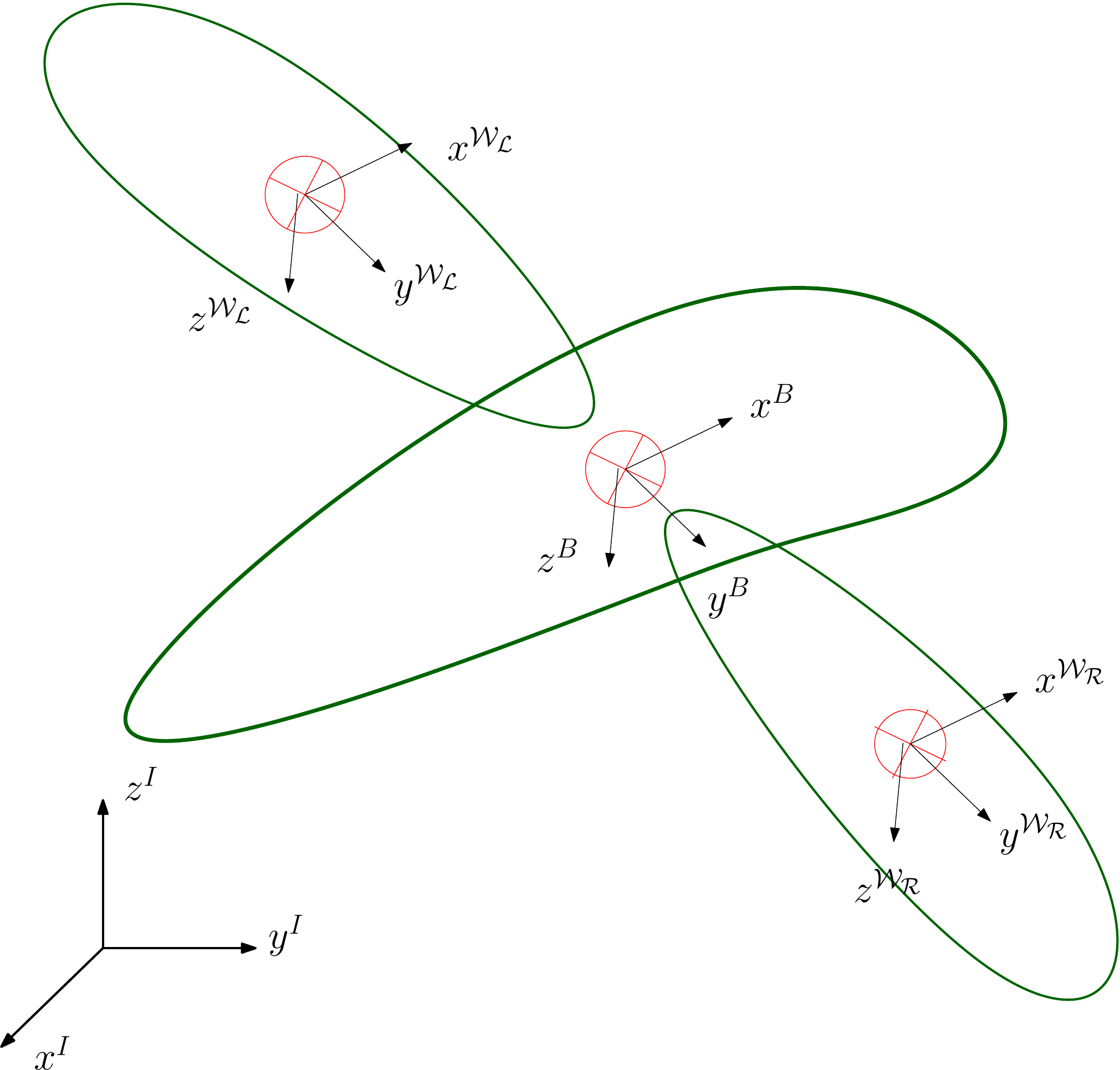}
\caption{Flapping winged vehicle schematic}
\label{fig:Flapping winged vehicle schematic}
\end{figure}

We shall refer to following four coordinate frames in modelling the dynamics of the rigid winged FWMAV - Inertial frame $\{x^\mathcal{I},y^\mathcal{I},z^\mathcal{I}\}$, body frame $\{x^\mathcal{B},y^\mathcal{B},z^\mathcal{B}\}$, left wing frame $\{x^\mathcal{W_L},y^\mathcal{W_L},z^\mathcal{W_L}\}$ and right wing frame $\{x^\mathcal{W_R},y^\mathcal{W_R},z^\mathcal{W_R}\}$. The origin of body and wing frames are at the respective centers of masses. The vehicle's longitudinal axis pointing forward is defined to be body x axis $x^\mathcal{B}$. The $y^\mathcal{B}$ axis points to the right of the body when looking in the direction of $x^\mathcal{B}$, and the $z^\mathcal{B}$ axis completes the right-handed frame. Wings' motion is of course subjected to the constraints by physical space which restrict allowable positions. For example, either of the wings can not rotate about $x^\mathcal{B}$ axis by $180$ degree, due to physical obstruction by the torso. 

\subsection{Configuration space}
Since the motion of the FWMAV is, in general, in 3 dimensional space, it can be represented using the special Euclidean Lie group $SE(3)$, which shall be the structure group for this problem. The shape variables in this case correspond to actuated degrees of freedom of 2 wings, each of whose motion evolves on a subset of  $SO(3)$. Thus we impart a principal fiber bundle topology to this configuration space. The notation is as follows:
\\ \\
\begin{tabularx}{\linewidth}{>{$}r<{$} @{${}\qquad{}$} X}
x_{\mathcal{B}},x_{\mathcal{W_L}},x_{\mathcal{W_R}} & Vectors from the origin of inertial frame to points on body, left wing and right wing, respectively, represented in inertial frame; \\
X_{\mathcal{B}},X_{\mathcal{W_L}},X_{\mathcal{W_R}} & Vectors from the origin of body frame, left wing frame, right wing frame to points on body, left wing and right wing, respectively, represented in respective frames; \\
r^\mathcal{B}_{\mathcal{BI}},r^\mathcal{I}_{\mathcal{BI}} & Vectors from the origin of inertial frame to the origin of body frame (torso of FMWAV), represented in body and inertial frames respectively; \\
r^\mathcal{I}_{\mathcal{W_L B}},r^\mathcal{W_L}_{\mathcal{W_L B}} & Vectors from the origin of body frame to the origin of left wing frame, represented in inertial and left frames respectively; \\
R_\mathcal{BI},R_\mathcal{W_LB} & Orientation of body frame with respect to inertial frame and left wing frame with respect to body frame; \\
m_\mathcal{B},m_\mathcal{W_L},m_\mathcal{W_R} & Mass of body (torso), left wing and right wing respectively; \\
I_\mathcal{B},I_\mathcal{W_L},I_\mathcal{W_R} & Moment of inertia of body (torso), left wing and right wing respectively; \\
\omega_{\mathcal{B}},\omega_{\mathcal{W_L}},\omega_{\mathcal{W_R}} & Angular velocities of body, left wing and right wing respectively with respect to inertial frame, represented in the respective frames; \\
\bar{h}_L, \bar{h}_R & The vector from the center of mass of body to the center of mass of left and right wings respectively; \\
\rho & The density of the body and wings, assumed to be constant across the entire FWMAV.
\end{tabularx}

\subsection{Lagrangian Computation}

We compute the Lagrangian $L$ for FWMAV in this section, given as the difference between kinetic energy $T$ and potential energy $V$

\begin{equation}\label{lagrangian}
L=T-V
\end{equation}

We compute the kinetic energy as the summation of that corresponding to body $(\mathcal{B})$, left wing $(\mathcal{W_L})$ and right wing $(\mathcal{W_R})$ of the FWMAV as follows.

\begin{align}
T&=T_\mathcal{B}+T_\mathcal{\mathcal{W_L}}+T_\mathcal{\mathcal{W_R}} \nonumber \\
&=\int_\mathcal{B} \norm{\dot{x}_\mathcal{B}}^2\rho \, d^3X_{\mathcal{B}} + \int_\mathcal{W_L} \norm{\dot{x}_\mathcal{W_L}}^2\rho \,  \mathrm d^3X_{\mathcal{W_L}} \nonumber \\ & \:\:\:\: + \int_\mathcal{W_R} \norm{\dot{x}_\mathcal{W_R}}^2\rho \, \mathrm d^3X_{\mathcal{W_R}} \nonumber
\end{align}
We now use the fact that a point on a body can be written with reference to the center of mass of the body. Also, we put to use the assumption that the wings are hinged at the center of mass of the torso. Hence,
\begin{align*}
&x_\mathcal{B}=r_\mathcal{BI}^{\mathcal{I}}+R_\mathcal{BI}X_{\mathcal{B}},~ \: x_\mathcal{W_L}=r_\mathcal{BI}^{\mathcal{I}}+R_\mathcal{BI}R_\mathcal{W_L B}X_{\mathcal{W_L}} ,\\
& x_\mathcal{W_R}=r_\mathcal{BI}^{\mathcal{I}}+R_\mathcal{BI}R_\mathcal{W_R B}X_{\mathcal{W_R}}
\end{align*}
We also use the fact that the body vector is constant in the body frame, hence
\begin{align}
T&=\int_\mathcal{B} \norm{\dot{r}_\mathcal{BI}^{\mathcal{I}}+\dot{R}_\mathcal{BI}X_{\mathcal{B}}}^2\rho \, \mathrm d^3X_{\mathcal{B}} + \nonumber
\\ & \:\:\:\: +\int_\mathcal{W_L} \norm{\dot{r}_\mathcal{BI}^{\mathcal{I}}+\dot{\overline{R_\mathcal{BI}R_\mathcal{W_L B}}X_\mathcal{W_L}}}^2\rho \, \mathrm d^3X_{\mathcal{W_L}} \nonumber
\\ & \:\:\:\: + \int_\mathcal{W_R} \norm{\dot{r}_\mathcal{BI}^{\mathcal{I}}+\dot{\overline{R_\mathcal{BI} R_\mathcal{W_R B}}}X_\mathcal{W_R}}^2\rho \, \mathrm d^3X_{\mathcal{W_R}} \label{kinetic}
\end{align}
The potential energy is also summation of that for body and both the wings,
\begin{equation}\nonumber
V=V_\mathcal{B}+V_\mathcal{W_L}+V_\mathcal{W_R}
\end{equation}
The potential energies of body, left wing and right wing are obtained as
\begin{align*}
&V_\mathcal{B}=m_\mathcal{B}g r^\mathcal{I}_{\mathcal{B}\mathcal{I}} \cdot z^{\mathcal{I}},~ V_\mathcal{W_L}=m_\mathcal{W_L}g r^\mathcal{I}_{\mathcal{W_L}\mathcal{I}} \cdot z^{\mathcal{I}},\\
&V_\mathcal{W_R}=m_\mathcal{W_R}g r^\mathcal{I}_{\mathcal{W_R}\mathcal{I}} \cdot z^{\mathcal{I}}
\end{align*}
where $g$ be the gravity and $z^{\mathcal{I}}$ be the unit vector in $z^{\mathcal{I}}$-direction, $r^\mathcal{I}_{\mathcal{W_L}}$ can be decomposed as 
\begin{align*}
r^\mathcal{I}_{\mathcal{W_L}\mathcal{I}} &= r^\mathcal{I}_{\mathcal{W_L}\mathcal{B}} + r^\mathcal{I}_{\mathcal{B}\mathcal{I}} \\
&= R_{\mathcal{BI}} r^\mathcal{B}_{\mathcal{W_L}\mathcal{B}} + r^\mathcal{I}_{\mathcal{B}\mathcal{I}}
= R_{\mathcal{B}\mathcal{I}} R_{\mathcal{W_LB}} \bar{h}_L + r^\mathcal{I}_{\mathcal{B}\mathcal{I}}.
\end{align*}
Hence the total potential energy is obtained as
\begin{align}\label{potential}
V&= m_{T}g r^\mathcal{I}_{\mathcal{B}\mathcal{I}} \cdot z^{I} + m_\mathcal{W_L}g(R_{\mathcal{B}\mathcal{I}} R_{\mathcal{W_LB}} \bar{h}_L) \cdot z^{\mathcal{I}} \nonumber \\
& \quad + m_\mathcal{W_R}g(R_{\mathcal{B}\mathcal{I}} R_{\mathcal{W_RB}} \bar{h}_R )\cdot z^{\mathcal{I}}
\end{align}
where $m_{T}= (m_\mathcal{B}+ m_{\mathcal{W_L}} + m_{\mathcal{W_R}})$. From (\ref{kinetic}) and (\ref{potential}) we have all the terms for finding Lagrangian in equation (\ref{lagrangian}).
\section{Group action, Lagrangian invariance and Equations of motion}
Symmetry is exploited to develop a dynamical model in a reduced space. A left action of a Lie group $G$ on manifold $Q$ is a mapping $\Phi : G \times Q \longrightarrow Q$ and the tangent lifted action by $T\Phi: G \times TQ \longrightarrow TQ $. Assuming the action is free and proper \cite{Marsden Ratiu intro to}, the Lagrangian $L$ is said to be invariant if it remains unchanged under the induced action of $G$ on $TQ$. The group under which $L$ is invariant is termed as the $symmetry$ group. To exploit symmetries in this problem, we cast the configuration space in the larger space, on the lines of \cite{holm} and then fix the gravitational vector for further analysis accordingly. Let configuration space $Q$ be the submanifold of the extended space $\tilde{Q} = Q \times \mathbb{R}^{3}$, so that the gravity vector $z^{\mathcal{I}}$ is considered as a new coordinate $p \in \mathbb{R}^{3}$. On this space define an extended Lagrangian $\tilde{L}$ be $\tilde{L}: T\tilde{Q} \longrightarrow \mathbb{R}$ which is invariant under the action $G$. Then the Lagrangian of the system is given by $L = \tilde{L}|_{p=z}$. The configuration space is identified as trivial fiber bundle with ($SO(3)$) as the bundle space or group space and ($ \mathbb{R}^{3} \times SO(3) \times SO(3)$) as the shape space. The left action of $G= SO(3)$ on $\tilde{Q}$ is given as
\begin{align*}
\Phi_{R_{1}}: (R_{\mathcal{BI}}, & r_{\mathcal{BI}}^\mathcal{I}, R_\mathcal{W_LB},R_\mathcal{W_RB},p) \\
& \longrightarrow (R_{1}R_{\mathcal{BI}}, R_{1}r_{\mathcal{BI}}, R_\mathcal{W_LB},R_\mathcal{W_RB},R_{1}p)
\end{align*}
and the tangent lifted action of $G$ is
\begin{align*}
T\Phi_{R_{1}}:  ( & \dot{R}_{\mathcal{BI}}, \dot{r}_{\mathcal{BI}}^\mathcal{I}, \dot{R}_\mathcal{W_LB},\dot{R}_\mathcal{W_RB},\dot{p}) \\
& \longrightarrow (R_{1}\dot{R}_{\mathcal{BI}}, R_{1}\dot{r}_{\mathcal{BI}}^\mathcal{I}, \dot{R}_\mathcal{W_LB},\dot{R}_\mathcal{W_RB},R_{1}\dot{p}).
\end{align*}
\begin{theorem}
The Lagrangian $L = \tilde{L}|_{p=z}$ is invariant under the group 
$$G_{z} = \{ R_{\mathcal{BI}} \in SO(3) | ~ R_{\mathcal{BI}}^{T}z^\mathcal{I} = z^\mathcal{I} \}.$$
\end{theorem}
\begin{proof}
Under the left action of $G$, the Lagrangian $L$ in (\ref{lagrangian}) is given by
\begin{align}\nonumber
\begin{split}
L &= \int_\mathcal{B} \norm{R_{1}(\dot{r}_\mathcal{BI}^{\mathcal{I}}+\dot{R}_\mathcal{BI}X_{\mathcal{B}})}^2\rho \, \mathrm d^3 X_\mathcal{B} + \\
& \quad +\int_\mathcal{W_L} \norm{R_{1}(\dot{r}_\mathcal{BI}^{\mathcal{I}}+\dot{\overline{R_\mathcal{BI} R_\mathcal{W_L B}}}X_\mathcal{W_L})}^2\rho \, \mathrm d^3 X_\mathcal{W_L} \\
& \quad + \int_\mathcal{W_R} \norm{R_{1}(\dot{r}_\mathcal{BI}^{\mathcal{I}}+\dot{\overline{R_\mathcal{BI} R_\mathcal{W_R B}}}X_\mathcal{W_R})}^2\rho \, \mathrm d^{3}X_\mathcal{W_R} \\
& \quad + m_{T}g \langle R_{1}r_\mathcal{BI}^{\mathcal{I}}, z^{\mathcal{I}} \rangle + m_{\mathcal{W_L}}g\langle R_{1}R_{\mathcal{BI}}r_\mathcal{W_{L}B}^{\mathcal{B}}, z^{\mathcal{I}} \rangle \\
& \quad + m_{\mathcal{W_R}}g\langle R_{1}R_{\mathcal{BI}}r_\mathcal{W_{R}B}^{\mathcal{B}}, z^{\mathcal{I}} \rangle ,
 \end{split}
\end{align}
\begin{align}
\begin{split}
& = \int_\mathcal{B} \norm{\dot{r}_\mathcal{BI}^{\mathcal{I}}+\dot{R}_\mathcal{BI}X_{\mathcal{B}}}^2\rho \, \mathrm d^3X_{\mathcal{B}} + \\
& \quad +\int_\mathcal{W_L} \norm{\dot{r}_\mathcal{BI}^{\mathcal{I}}+ \dot{\overline{R_\mathcal{BI} R_\mathcal{W_L B}}}X_\mathcal{W_L}}^2\rho \, \mathrm d^3X_{\mathcal{W_L}} \\
& \quad + \int_\mathcal{W_R} \norm{\dot{r}_\mathcal{BI}^{\mathcal{I}}+ \dot{\overline{R_\mathcal{BI} R_\mathcal{W_R B}}} X_\mathcal{W_R}}^2\rho \, \mathrm d^3X_{\mathcal{W_R}} \\
& \quad + m_{T}g \langle r_\mathcal{BI}^{\mathcal{I}}, R_{1}^{T}z^{\mathcal{I}} \rangle + m_{\mathcal{W_L}}g\langle R_{\mathcal{BI}}r_\mathcal{W_{L}B}^{\mathcal{B}}, R_{1}^{T}z^{\mathcal{I}} \rangle \\
& \quad + m_{\mathcal{W_R}}g\langle R_{\mathcal{BI}}r_\mathcal{W_{R}B}^{\mathcal{B}}, R_{1}^{T}z^{\mathcal{I}} \rangle
\end{split}
\end{align}
where the fact that $R_{1}R_{1}^{T}=I$ has been used. Since $R_{1}^{T}z^{\mathcal{I}}=z^{\mathcal{I}}$, the Lagrangian remains the same.
\end{proof}
\textit{Remark}: The gravity breaks the full group symmetry, and hence leaving symmetry with respect to subgroup $G_{z}=SO(2)$, a subgroup of $SO(3)$. In other words, the system Lagrangian remain unchanged if we rotate about $z^{\mathcal{I}}$ axis. From the argument, we see that $(Q,L)$ can be constructed by $(\tilde{Q},\tilde{L})$.

When the Lagrangian $L$ is invariant under the group action the equations get reduced to the quotient space $TQ/G_{z}$. Let $\mathfrak{so}(3)$ be the Lie algebra of $G$ and $\mathcal{O}$ be the orbit space (Appendix C) $G/G_{z}$ of $z^{\mathcal{I}}$ in $\mathbb{R}^{3}$, then we have reduced the system from $T\tilde{Q}$ to $\mathfrak{so}(3) \times \mathcal{O} \times T\mathbb{R}^{3} \times TSO(3) \times TSO(3)$. 

Given a curve $(R_{\mathcal{BI}}(t), r_{\mathcal{BI}}^{\mathcal{I}}(t), R_{\mathcal{W_LB}}(t), R_{\mathcal{W_RB}}(t)) \in Q$, the $L$ is extracted from extended Lagrangian $\tilde{L}$ as
\begin{align*}\nonumber
&L( R_{\mathcal{BI}}, r_{\mathcal{BI}}^{\mathcal{I}}, R_{\mathcal{W_LB}}, R_{\mathcal{W_RB}},\dot{R}_{\mathcal{BI}},\dot{r}_{\mathcal{BI}}^{\mathcal{I}},\dot{R}_{\mathcal{W_LB}},\dot{R}_{\mathcal{W_RB}})\\
& = \tilde{L}(R_{\mathcal{BI}}, r_{\mathcal{BI}}^{\mathcal{I}}, R_{\mathcal{W_LB}}, R_{\mathcal{W_RB}},\dot{R}_{\mathcal{BI}},\dot{r}_{\mathcal{BI}}^{\mathcal{I}},\dot{R}_{\mathcal{W_LB}},\dot{R}_{\mathcal{W_RB}},z^{\mathcal{I}}).
\end{align*}
Define the body(torso)-coordinate as
\begin{equation}\label{body_coordinates}
r_{\mathcal{BI}}^{B}=R_{\mathcal{BI}}^{T}r^\mathcal{I}_{\mathcal{BI}}, \quad \widehat{\omega}_{\mathcal{B}}= R_{\mathcal{BI}}^{T}\dot{R}_{\mathcal{BI}}, \quad \Gamma = R_{\mathcal{BI}}^{T}z^{\mathcal{I}}
\end{equation}
where $\widehat{\omega}_{\mathcal{B}} \in \mathfrak{so}(3)$ is the angular velocity in torso-body-coordinate, $r_{\mathcal{BI}}^{B} \in \mathbb{R}^{3}$ is the position vector in torso-body frame and $\Gamma \in \mathbb{R}^{3}$ is called an advected vector \cite{cendra}. Then the reduced Lagrangian $l: (\mathfrak{so}(3) \times \mathcal{O} \times
T\mathbb{R}^{3} \times TSO(3) \times TSO(3)) \longrightarrow \mathbb{R}$ is given by
\begin{align*}
L&( R_{\mathcal{BI}}, r_{\mathcal{BI}}^{\mathcal{I}}, R_{\mathcal{W_LB}}, R_{\mathcal{W_RB}},\dot{R}_{\mathcal{BI}},\dot{r}_{\mathcal{BI}}^{\mathcal{I}},\dot{R}_{\mathcal{W_LB}},\dot{R}_{\mathcal{W_RB}})\\
&=l(e,R_{\mathcal{BI}}^{T}r_{\mathcal{BI}}^{\mathcal{I}}, R_{\mathcal{W_LB}}, R_{\mathcal{W_RB}},R_{\mathcal{BI}}^{T}\dot{R}_{\mathcal{BI}},\dot{r}_{\mathcal{BI}}^{\mathcal{I}},\dot{R}_{\mathcal{W_LB}},\\ & \qquad \dot{R}_{\mathcal{W_RB}},R_{\mathcal{BI}}^{T}z^{\mathcal{I}}),\\
& = l(r_{\mathcal{BI}}^{\mathcal{B}},R_{\mathcal{W_LB}}, R_{\mathcal{W_RB}},\widehat{\omega}_{\mathcal{B}},\dot{r}_{\mathcal{BI}}^{B},\dot{R}_{\mathcal{W_LB}},\dot{R}_{\mathcal{W_RB}},\Gamma).
\end{align*}
We use (\ref{body_coordinates}) and allude to Appendix D for kinetic energy calculations to get the following expression for the reduced Lagrangian.
\begin{align}\label{red_lag}
\begin{split}
l & =\frac{1}{2}m_{\mathcal{B}}\norm{\dot{r}_{\mathcal{BI}}^{\mathcal{B}}}^2+\frac{1}{2}\omega_{\mathcal{B}}^{T} I_{\mathcal{B}} \omega_{\mathcal{B}} + \frac{1}{2}m_{\mathcal{W_L}}\norm{\dot{r}_{\mathcal{BI}}^{\mathcal{B}}}^2\\
& \quad \quad + \frac{1}{2}\omega_{\mathcal{B}}^{T}(R_\mathcal{W_LB}I_{\mathcal{W_L}}R_\mathcal{W_LB}^T)\omega_{\mathcal{B}} + \frac{1}{2}\omega_{\mathcal{W_L}}^{T}I_{\mathcal{W_L}}\omega_{\mathcal{W_L}} \\
& \quad \quad + \omega_{\mathcal{B}}^{T} (R_\mathcal{W_LB}I_{\mathcal{W_L}} R_\mathcal{W_LB}) \omega_{\mathcal{W_L}} + \frac{1}{2}m_{\mathcal{W_R}}\norm{\dot{r}_{\mathcal{BI}}^{\mathcal{B}}}^2 \\
&\quad \quad +\frac{1}{2}\omega_{\mathcal{B}}^T(R_\mathcal{W_RB}I_{\mathcal{W_R}}R_\mathcal{W_RB}^T)\omega_{\mathcal{B}} + \frac{1}{2}\omega_{\mathcal{W_R}}^{T}I_{\mathcal{W_R}}\omega_{\mathcal{W_R}} \\
& \quad \quad + \omega_{\mathcal{B}} (R_\mathcal{W_RB}I_{\mathcal{W_R}} R_\mathcal{W_RB}) \omega_{\mathcal{W_R}} + V( r_{\mathcal{BI}}^{\mathcal{B}},\Gamma), \\
&= Z M Z^T + V( r_{\mathcal{BI}}^{\mathcal{B}},\Gamma)
\end{split}
\end{align}
\begin{align*}
\text{where,}\:\:\: & Z=\begin{bmatrix}
r_{\mathcal{BI}}^{\mathcal{B}} & \omega_{\mathcal{B}} & \omega_{\mathcal{W_L}} & \omega_{\mathcal{W_R}}
\end{bmatrix} \\
& V = m_{T}g \langle r_{\mathcal{BI}}^{\mathcal{B}}, \Gamma \rangle + m_{W_L}g \langle R_{\mathcal{W_LB}} \bar{h}_L ,\Gamma \rangle \\
& \quad \quad + m_{W_R}g \langle R_{\mathcal{W_RB}} \bar{h}_R ,\Gamma \rangle
\end{align*}
\begin{align}\nonumber
\begin{split}
M =\begin{bmatrix}
M_{1}\\
M_{2}\\
M_{3}\\
M_{4}
\end{bmatrix}=\begin{bmatrix}
m_{11} & m_{12} & m_{13} & m_{14} \\
m_{21} & m_{22} & m_{23} & m_{24} \\
m_{31} & m_{32} & m_{33} & m_{34} \\
m_{41} & m_{42} & m_{43} & m_{44} 
\end{bmatrix}
\end{split}
\end{align}
\begin{align}\nonumber
\begin{split}
& m_{11} = m_{T}I_{3 \times 3}; ~ m_{33} = I_{W_L}; ~ m_{44} = I_{W_R}, \\
& m_{22}=I_\mathcal{B} + R_\mathcal{W_LB}I_{W_L}R_\mathcal{W_LB}^T + R_\mathcal{W_RB}I_{W_R}R_\mathcal{W_RB}^T, \\
& m_{23} = R_\mathcal{W_LB}I_{W_L}R_\mathcal{W_LB}^T = m_{32},\\
& m_{24} = R_\mathcal{W_RB}I_{W_R}R_\mathcal{W_RB}^T = m_{42}.
\end{split}
\end{align}

The other terms are zero.

We now calculate the dynamic equation of the reduced Lagrangian $\mathit{l}$ using Hamilton's variational principle. Let $\delta$ denote the variational operator. The dynamic equations are given by solving the following equation -
\begin{equation}\nonumber
\delta \int \limits^{t_{1}}_{t_{0}} l(r_{\mathcal{BI}}^{\mathcal{B}},R_{\mathcal{W_LB}}, R_{\mathcal{W_RB}},\widehat{\omega}_{\mathcal{B}},\dot{r}_{\mathcal{BI}}^{B},\dot{R}_{\mathcal{W_LB}},\dot{R}_{\mathcal{W_RB}},\Gamma)\, \mathrm dt=0,
\end{equation}
where the curve $R_{\mathcal{BI}}(t)$ joins two fixed configurations in $SO(3)=G$ and the curves $r_{\mathcal{BI}}^{I}(t)$, $R_{\mathcal{W_LB}}(t)$ and $R_{\mathcal{W_RB}}(t)$ joins two fixed configurations in $\mathbb{R}^{3} \times SO(3) \times SO(3) = Q_{s}$ and hence covers the entire space $Q=G \times Q_{s}$. The variations $\delta R_{\mathcal{BI}}(t)$, $\delta r_{\mathcal{BI}}^{\mathcal{I}}(t)$, $\delta R_{\mathcal{W_LB}}(t)$ and $\delta R_{\mathcal{W_RB}}(t)$ are independent variations vanishing at the end points, hence,
$$\delta R_{\mathcal{BI}}(t_{0})= R_{\mathcal{BI}}(t_{1})=0,~~ \delta R_{\mathcal{W_LB}}(t_{0})=\delta R_{\mathcal{W_LB}}(t_{1})=0.$$ $$\delta R_{\mathcal{W_RB}}(t_{0})=\delta R_{\mathcal{W_RB}}(t_{1})=0~~\hbox{and}~~\delta r_{\mathcal{BI}}^{\mathcal{B}} (t_{0})= \delta r_{\mathcal{BI}}^{\mathcal{B}} (t_{1})=0,$$
So, the explicit variation of the integral is
\begin{align}
& \delta  \int \mathit{l}\, \mathrm dt = \int \left\{ \left\langle \frac{\partial \mathit{l}}{\partial \omega_{\mathcal{B}}},\delta \omega_{\mathcal{B}}\right\rangle + \left\langle  \frac{\partial \mathit{l}}{\partial r_{\mathcal{BI}}^{\mathcal{B}}},\delta r_{\mathcal{BI}}^{\mathcal{B}}\right\rangle + \left\langle  \frac{\partial \mathit{l}}{\partial \Gamma},\delta \Gamma \right\rangle  \right. \nonumber \\
&  + \left. \left\langle  \frac{\partial \mathit{l}}{\partial \dot{r}_{\mathcal{BI}}^{\mathcal{B}}},\delta \dot{r}_{\mathcal{BI}}^{\mathcal{B}}\right\rangle + \left\langle  \frac{\partial \mathit{l}}{\partial R_{\mathcal{W_LB}}},\delta R_{\mathcal{W_LB}} \right\rangle + \left\langle  \frac{\partial \mathit{l}}{\partial R_{\mathcal{W_RB}}},\delta R_{\mathcal{W_RB}} \right\rangle \right. \nonumber 
\\
&  + \left.  \left\langle  \frac{\partial \mathit{l}}{\partial \dot{R}_{\mathcal{W_LB}}},\delta \dot{R}_{\mathcal{W_LB}} \right\rangle + \left\langle  \frac{\partial \mathit{l}}{\partial \dot{R}_{\mathcal{W_RB}}},\delta \dot{R}_{\mathcal{W_RB}}\right\rangle  \right\} \, \mathrm dt. \label{variation}
\end{align}
To compute the variations of $\delta \omega_{\mathcal{B}},\delta r_{\mathcal{BI}}^{\mathcal{B}}$, $\delta \dot{r}_{\mathcal{BI}}^{\mathcal{B}}$ and $\delta \Gamma $, we proceed as follows. Let $\eta$ and $\bar{y}$ be the functions satisfying $\widehat{\eta}= R_{\mathcal{BI}}^{T}\delta R_{\mathcal{BI}}$ and $\bar{y}=R_{\mathcal{BI}}^{T} \delta r_{\mathcal{BI}}^{\mathcal{I}}$ respectively, vanishing at the end points. The variations $\widehat{\omega}_{\mathcal{B}}$, $\delta r_{\mathcal{BI}}^{\mathcal{B}}$, $\delta \dot{r}_{\mathcal{BI}}^{\mathcal{B}}$ and $\delta \Gamma$ are computed as
$$ \delta \omega_{\mathcal{B}} = \dot{\eta} + ad_{\omega_{\mathcal{B}}}\eta, ~~ \delta r_{\mathcal{BI}}^{\mathcal{B}} =-\widehat{\eta}r_{\mathcal{BI}}^{\mathcal{B}} +\bar{y}$$
$$\delta \dot{r}_{\mathcal{BI}}^{\mathcal{B}} =-\widehat{\eta}\dot{r}_{\mathcal{BI}}^{\mathcal{B}} +\widehat{\omega}_{\mathcal{B}}\bar{y} +\dot{\bar{y}} $$ and $\quad \delta \Gamma = -R_{\mathcal{BI}}^{-1}\delta R_{\mathcal{BI}}R_{\mathcal{BI}}^{-1} z^{\mathcal{I}}= -\widehat{\eta}\Gamma.$

By substituting these in equation (\ref{variation}), applying Hamilton's variational principle, we get following equations representing the dynamics of the FWMAV -
\begin{align}
& \frac{d}{dt}\left(\dfrac{\partial l}{\partial \dot{r}_{\mathcal{BI}}^{\mathcal{B}}}\right) - \frac{\partial l}{\partial r_{\mathcal{BI}}^{\mathcal{B}}} = \left(\frac{\partial l}{\partial \dot{r}_{\mathcal{BI}}^{\mathcal{B}}}\times \omega_{\mathcal{B}} \right) \nonumber \\
& \frac{d}{dt}  \left(\frac{\partial l}{ \partial \omega_{\mathcal{B}}}\right) - ad_{\omega_{\mathcal{B}}}^{\ast} \left(\frac{\partial l}{\partial \omega_{\mathcal{B}}} \right) =  \left(\frac{\partial l}{\partial \dot{r}_{\mathcal{BI}}^{\mathcal{B}}} \times \dot{r}_{\mathcal{BI}}^{\mathcal{B}} \right) + \left( r_{\mathcal{BI}}^{\mathcal{B}} \times \frac{\partial l}{\partial r_{\mathcal{BI}}^{\mathcal{B}}} \right) \nonumber \\
& \qquad \qquad\qquad\qquad\qquad -\left(\frac{\partial l}{\partial \Gamma} \times \Gamma \right) \nonumber \\
& \frac{d}{dt}\left(\frac{\partial l}{\partial \dot{R}_{\mathcal{W_LB}}} \right) - \frac{\partial l}{\partial R_{\mathcal{W_LB}}} = 0, \nonumber \\
& \frac{d}{dt}\left(\frac{\partial l}{\partial \dot{R}_{\mathcal{W_RB}}}\right) - \frac{\partial l}{\partial R_{\mathcal{W_RB}}} = 0 \label{dynamics}
\end{align}

We substitute the reduced Lagrangian in above equations and compute the derivatives. After inclusion of the aerodynamic and control forces we finally get the equations of motion for the FWMAV as follows

\begin{align}\label{final_dyn}
\begin{split}
& \frac{d}{dt} \left( M_{1} Z \right) =  (m_{B} + m_{W_L} + m_{W_R}) \dot{r}_{\mathcal{BI}}^{B}\times \omega_\mathcal{B} + F^{a}, \\
& \frac{d}{dt} \left( M_{2} Z \right) = M_{1}Z \times \omega_{\mathcal{B}} + m_{W_{L}}g \bar{h}_L \times \Gamma + m_{W_{R}}g \bar{h}_R \times \Gamma +T^{a}_\mathcal{B},\\
& \frac{d}{dt} \left( M_{3} Z \right) - \left( \frac{\partial l}{\partial R_{\mathcal{W_LB}}} \right) = T^{c}_\mathcal{W_L}+T^{a}_\mathcal{W_L}, \\
& \frac{d}{dt} \left( M_{4} Z \right) - \left( \frac{\partial l}{\partial R_{\mathcal{W_RB}}} \right) = T^{c}_\mathcal{W_R}+T^{a}_\mathcal{W_R}
\end{split}
\end{align}

where,
\vspace{5pt}
\\
\begin{tabularx}{\linewidth}{>{$}r<{$} @{${}\:\:\:{}$} X}
T^{a}_\mathcal{B},T^{a}_\mathcal{W_L}, T^{a}_\mathcal{W_R}  \: :& Aerodynamic torques acting on the body, left wing and right wing respectively; \\
F^{a} \: :& Total aerodynamic force acting on FWMAV; \\
T^{c}_\mathcal{W_L}, T^{c}_\mathcal{W_R} \: :&  Control torques applied on the left and right wings respectively, at their respective joints with the body.
\end{tabularx}
\\
\\
The first two equations of (\ref{final_dyn}) are the Euler-Poincar\'{e} equations whereas the last two are the Euler-Lagrange equations. The evolution of $\Gamma$ is calculated from the advected dynamics as
\begin{equation}\label{adv_dynamics}
\dot{\Gamma} = - \omega_{B} \times \Gamma.
\end{equation}
Using the solution $\omega_{\mathcal{B}}$ of the equation (\ref{final_dyn}), we can now find the curve $R_{\mathcal{BI}}(t)$ by reconstruction equation
\begin{equation}\label{reconstruct}
\dot{R}_{\mathcal{BI}}(t)=R_{\mathcal{BI}}(t) \widehat{\omega}_{\mathcal{B}} \mbox{   with } R_{\mathcal{BI}}(0) = R_{\mathcal{BI}_{0}}.
\end{equation}
Hence, equation (\ref{final_dyn}), (\ref{adv_dynamics}) together with the reconstruction equation (\ref{reconstruct}), give the complete dynamics of the FWMAV.
\section{CONCLUSIONS and FUTURE WORK}
This work is intended to explore the application of geometric mechanics to locomotion on $SE(3)$ for the FWMAV problem. The existing literature on the topic relies on traditional techniques based on classical mechanics and local parametrization, which tends to hide the underlying geometric structures of the system dynamics. This work has achieved to find the equations of motion in a coordinate a free setting.

In this paper we have incorporated the aerodynamic forces and moments without much regard to their form. It is known that accurate modelling of these forces renders accounting for complex phenomena. High-fidelity models are generally not amenable to analytical control theoretic methods. Instead, low-fidelity models, which although would not fully represent the physics, can help serve useful purpose. One is that it allows a description of the system in a control affine form in which the control inputs enter linearly. The reference \cite{Morgansen UUV} does such analysis for the hydrodynamic locomotion problem. It also gives controllability analysis for the underwater vehicle, using Lie bracket conditions. Such analysis can also be performed for the FWMAV. Feasibility of range of interesting manoeuvres exhibited by insects and MAVs can be studied through such a controllability analysis.

It also appears that in addition to the $SO(2)$ symmetry the FWMAV has $\mathbb{R}^2$ symmetry, which corresponds to the motion in the $x-y$ plane. Under the action of this $\mathbb{R}^2$ group the potential and kinetic energies are invariant. Hence a reduction under a semidirect product group $SE(2)= SO(2) \circledS \mathbb{R}^2$ can be explored.


\section*{APPENDIX}

\subsection{Trivial principal fiber bundle}\label{Trivial principal fiber bundle}
Let $Q$ be a manifold and $G$ a Lie group. A trivial principal fiber bundle with base $B$ and structure group $G$ is a manifold $Q = B \times G$ with a free left action of $G$ on $Q$ given by left translation in the group variable : $\phi_h(x,g) = (x,hg)$ for $x \in M$ and $g \in G$.

\subsection{Quotient space}\label{Quotient space}
The quotient space $X/ \sim$ of a topological space $X$ and an equivalence relation $\sim$ on $X$ is the set of equivalence classes of points in $X$, under the equivalence relation $\sim$.

\subsection{Orbit space}\label{Orbit space}
We identify $SO(3)/SO(2)$ as the orbit space $\mathcal{O}$ in $\mathbb{R}^{3}$ defined by
$$ \mathcal{O} = \{ y \in \mathbb{R}^{3} | y= gz \mbox{ for some } g\in  SO(3) \}. $$
If $\mathcal{O}$ is closed in $\mathbb{R}^{3}$, then $SO(3)/SO(2)$ is diffeomorphic to $\mathcal{O}$, a submanifold of $\mathbb{R}^{3}.$

\subsection{Explicit computation of equation (\ref{red_lag}) }\label{KE computation}
The reduced Lagrangian is given as,
\begin{align*}
l&  = \int_\mathcal{B}\norm{\dot{r}_{\mathcal{BI}}^{B}+\widehat{\omega}_\mathcal{B}X_{\mathcal{B}}}^2\rho \, \mathrm d^3X_{\mathcal{B}} \\
& \quad + \int_\mathcal{W_L} \norm{\dot{r}_\mathcal{BI}^{\mathcal{B}}+ (\widehat{\omega}_\mathcal{B} R_{\mathcal{W_LB}} + R_{\mathcal{W_LB}} \widehat{\omega}_{W_L}) X_\mathcal{W_L}}^2\rho \, \mathrm d^3X_{\mathcal{W_L}} \\
& \quad + \int_\mathcal{W_R} \norm{\dot{r}_\mathcal{BI}^{\mathcal{B}}+ (\widehat{\omega}_\mathcal{B} R_{\mathcal{W_LB}} + R_{\mathcal{W_RB}} \widehat{\omega}_{W_R}) X_\mathcal{W_R}}^2\rho \, \mathrm d^3X_{\mathcal{W_R}} \\
& \quad - V(r_{\mathcal{BI}}^{\mathcal{B}})
\end{align*}
We present the computation of integration for kinetic energy calculation of the body. 
\begin{align*}
\int_\mathcal{B} & \norm{\dot{r}_{\mathcal{BI}}^{B}+\widehat{\omega}_\mathcal{B}X_{\mathcal{B}}}^2\rho \,d^3X_{\mathcal{B}}\:
\\
&=\:\int_\mathcal{B}\left( \norm{\dot{r}_{\mathcal{BI}}^{B}}^2+\norm{\omega_\mathcal{B}\times X}^2 + 2 \, \langle \dot{r}_{\mathcal{BI}}^{B},\widehat{\omega}_\mathcal{B} X_{\mathcal{B}} \rangle \right)  \, \rho \, d^3X_{\mathcal{B}} \\
&=\norm{\dot{r}_{\mathcal{BI}}^{\mathcal{B}}}^2 \int_\mathcal{B}\rho \, d^3X_{\mathcal{B}} + \omega_\mathcal{B} \cdot \left( \int_\mathcal{B} \rho \hat{X}_{\mathcal{B}}^T \hat{X}_{\mathcal{B}} \, d^3X_{\mathcal{B}} \right)  \omega_\mathcal{B}
\\
&=m_\mathcal{B} \norm{\dot{r}_{\mathcal{BI}}^{\mathcal{B}}}^2 + (\omega_\mathcal{B})^TI_\mathcal{B}\omega_\mathcal{B}
\end{align*}
where, the body mass, body inertia and position of center of mass in body coordinates is given respectively as
\begin{align*}
&m_\mathcal{B}=\int_\mathcal{B} \rho \, d^3X_{\mathcal{B}}; \: I_\mathcal{B}=\int_\mathcal{B}\widehat{X}_{\mathcal{B}}^{T}\widehat{X}_{\mathcal{B}}\rho \, d^3X_{\mathcal{B}}
\end{align*}
The calculation for wings shall be similar, hence the reduced kinetic energy is given as
\begin{align*}
T_{red}&  = m_\mathcal{B} \norm{\dot{r}_{\mathcal{BI}}^{\mathcal{B}}}^2 + (\omega_\mathcal{B})^{T}I_\mathcal{B}\omega_\mathcal{B} + \frac{1}{2}m_{W_L}\norm{\dot{r}_{\mathcal{BI}}^{\mathcal{B}}}^2\\
& \quad \quad + \frac{1}{2}\omega_{\mathcal{B}}^{T}(R_{\mathcal{W_LB}}I_{W_L}R_{\mathcal{W_LB}}^T)\omega_{\mathcal{B}} + \frac{1}{2}\omega_{\mathcal{W_L}}^{T}I_{\mathcal{W_L}}\omega_{\mathcal{W_L}} \\
& \quad \quad + \omega_{\mathcal{B}}^{T} (R_{\mathcal{W_LB}}I_{W_L} R_{\mathcal{W_LB}}) \omega_{\mathcal{W_L}} + \frac{1}{2}m_{W_R}\norm{\dot{r}_{\mathcal{BI}}^{\mathcal{B}}}^2 \\
&\quad \quad +\frac{1}{2}\omega_{\mathcal{B}}^T(R_{\mathcal{W_RB}}I_{W_R}R_{\mathcal{W_RB}}^T)\omega_{\mathcal{B}} + \frac{1}{2}\omega_{W_R}^{T}I_{\mathcal{W_R}}\omega_{\mathcal{W_R}} \\
& \quad \quad + \omega_{\mathcal{B}} (R_{\mathcal{W_RB}}I_{W_L} R_{\mathcal{W_RB}}) \omega_{\mathcal{W_R}}
\end{align*}
where 
\begin{align*}
& m_\mathcal{W_L}=\int_\mathcal{B} \rho \, d^3X_{\mathcal{W_L}}; \:\:\: I_\mathcal{W_L}=\int_\mathcal{W_L}\widehat{X}_{\mathcal{W_L}}^T\widehat{X}_{\mathcal{W_L}}\rho \, d^3X_{\mathcal{W_L}}, \\
& m_\mathcal{W_R}=\int_\mathcal{W_R} \rho \, d^3X_{\mathcal{W_R}}; \:\:\: I_\mathcal{W_R}=\int_\mathcal{W_R}\widehat{X}_{\mathcal{W_R}}^T\widehat{X}_{\mathcal{W_R}}\rho \, d^3X_{\mathcal{W_R}}.
\end{align*}



\begin{thebibliography}{99}

\bibitem{kelly murray} S. D. Kelly, and R. M. Murray, Geometric phases and robotic locomotion, Journal of Robotic Systems, vol. 12, no. 6, pp. 417-431, 1995.

\bibitem{liang} Y. Liang, A Review of Geometry in Robotic Locomotion, Technical project, California Institute of Technology, 1996.

\bibitem{cabrera} A. Cabrera, Base-controlled mechanical systems and geometric phases, Journal of Geometry and Physics, vol. 58, no. 3, pp. 334-367, 2008.

\bibitem{Bullo project} F. Bullo, On controllability and symmetries in simple mechanical systems, Technical project, California Institute of Technology, 1996.

\bibitem{Fairchild beanie} S. D. Kelly, M. J. Fairchild, P. M. Hassing, and P. Tallapragada, Proportional heading control for planar navigation: The Chaplygin beanie and fishlike robotic swimming, Americal Control Conference, Canada, June 2012.

\bibitem{Chung dorothy} S. Chung, and  M. Dorothy, Neurobiologically inspired control of engineered flapping flight, Journal of guidance, control, and dynamics, vol. 33, no. 2, pp. 440-453, March-April 2010.

\bibitem{Taha thesis} H. Taha, Mechanics of Flapping Flight: Analytical Formulations of Unsteady Aerodynamics, Kinematic Optimization, PhD dissertation, Virginia Polytechnic Institute and State University, Oct. 2013.

\bibitem{Taha Hajj} H. Taha, M. R. Hajj, and P. S. Beran, State-space representation of the unsteady aerodynamics of flapping flight, Aerospace Science and Technology 34, 2014.

\bibitem{Kuang robobat} P. D. Kuang, M. Dorothy, and S. Chung, Robobat: Dynamics and control of a robotic bat flapping flying testbed, AIAA Infotech at Aerospace Conference, St. Louis, MO. March 2011.

\bibitem{paranjape review paper} A. Paranjape, M. Dorothy, S. Chung, and K. D. Lee, A flight mechanics-centric review of bird-scale flapping flight. International Journal of Aeronautical and Space Sciences, vol. 13, no. 3, pp. 267-281, 2012.

\bibitem{Taha review} H. E. Taha, M. R. Hajj, and H. N. Ali, Flight dynamics and control of flapping-wing MAVs: a review. Nonlinear Dynamics, vol. 70, no. 2, pp. 907-939, 2012.

\bibitem{Morgansen UUV} K. Morgansen, B. Triplett, and D. J. Klein, Geometric methods for modeling and control of free-swimming fin-actuated underwater vehicles. Robotics, IEEE Transactions on robotics, vol. 23, no. 6, pp. 1184-1199, 2007.

\bibitem{Bullo book} F. Bullo, Geometric control of mechanical systems. vol. 49. Springer Science \& Business Media, 2005.

\bibitem{Bloch book} A. M. Bloch, Nonholonomic mechanics and control. vol. 24. Springer Science \& Business Media, 2003.

\bibitem{Kobayashi} S. Kobayashi, and K. Nomizu, Foundations of differential geometry. vol. 1, Interscience Publishers, 1963.

\bibitem{Paranjape Chung tailless aircraft} A. Paranjape, S. J. Chung, and M. S. Selig, Flight mechanics of a tailless articulated wing aircraft, Bioinspiration \& Biomimetics, vol. 6, no. 2, 2011.

\bibitem{Marsden Ratiu intro to} J. E. Marsden, and T. Ratiu, Introduction to mechanics and symmetry: a basic exposition of classical mechanical systems. Vol. 17. Springer Science \& Business Media, 2013.

\bibitem{cendra} H. Cendra, D. D. Holm, J. E. Marsden, and T. S. Ratiu. Lagrangian reduction, the Euler-poincare equations, and semidirect products, American Mathematical Society Translation(2), volume~186, 1998.

\bibitem{marsden_krishna_bloch}
A. M. Bloch, P. S. Krishnaprasad, J. E. Marsden, and R. M. Murray, Nonholonomic Mechanical systems with symmetry,
Archive for Rational Mechanics and Analysis, Springer-Verlag, volume 136,1996.

\bibitem{holm}
D. Holm, T. Schmah, and C. Stoica, Geometric Mechanics and Symmetry, Oxford University Press Inc., New York, 2009.

\bibitem{gajbhiye_banavar}
S. Gajbhiye, and R. N. Banavar, Euler-Poincare Equations for a spherical robot actuated by a pendulum, Proceedings of 4th IFAC Workshop on Lagrangian and Hamiltonian methods for Non Linear Control, vol. 4, pp 72-77, 2012.

\bibitem{Gallier}
J. Gallier, Geometric methods and applications: for computer science and engineering. Vol. 38. Springer Science \& Business Media, 2011.

\bibitem{Wilson}
J. Wilson, Manifold Theory. Lecture notes. Workshop of Multiagent Pathfinding (WOMP), University of Chicago, 2012.

\bibitem{ostrowski}
J. Ostrowski, J. Burdick, The geometric mechanics of undulatory robotic locomotion, The international journal of robotics research, vol. 17, no. 7,  1998.

\bibitem{wolfram_website}
T. Rowland, Fiber Bundle. From MathWorld--A Wolfram Web Resource, created by Eric W. Weisstein., [online] http://mathworld.wolfram.com/FiberBundle.html. [Accessed : 20 - Aug - 2015]

\end{thebibliography}
\end{document}